\documentclass[a4paper,UKenglish,cleveref, autoref, thm-restate,nolineno]{socg-lipics-v2021}

\title{Recognizing Subgraphs of Regular Tilings}

\author{Eliel Ingervo}
    {Aalto University, Finland}
    {eliel.ingervo@aalto.fi}
    {https://orcid.org/0009-0007-0494-979X}{}

\author{S\'andor Kisfaludi-Bak}
    {Aalto University, Finland}
    {sandor.kisfaludi-bak@aalto.fi}
    {https://orcid.org/0000-0002-6856-2902}{}

\funding{This work was supported by the Research Council of Finland, Grant 363444.}

\authorrunning{E. Ingervo and S. Kisfaludi-Bak}

\hideLIPIcs

\ccsdesc[300]{Theory of computation~Computational geometry}
\ccsdesc[300]{Theory of computation~Graph algorithms analysis}
\ccsdesc[300]{Theory of computation~Dynamic programming}

\keywords{Subgraph isomorphism, Sphere cut decomposition, Regular tiling, Hyperbolic tiling} 

\usepackage[utf8]{inputenc}
\usepackage[T1]{fontenc}
\usepackage{amsmath,amssymb,amsthm}
\usepackage[dvipsnames]{xcolor}
\usepackage[linesnumbered,ruled,vlined]{algorithm2e}
\usepackage{mathtools}
\usepackage{url}
\usepackage{thm-restate}
\usepackage{cleveref}
\usepackage{graphicx}
\usepackage{thm-restate}

\SetKw{Continue}{continue}

\graphicspath{{Figures/}}

\newtheorem{question}{Question}

\newcommand{\tiling}{{\mathcal{T}_{p,q}}}
\newcommand{\til}{{T_{p,q}}}
\newcommand{\conv}{\mathrm{conv}}
\newcommand{\reg}{\mathrm{Reg}}

\newcommand{\tB}{\widetilde B}
\newcommand{\wt}{\widetilde}
\newcommand{\bp}{\overline\phi}
\newcommand{\dom}{\mathrm{Dom}}
\newcommand{\ba}[1]{\overline{#1}}
\newcommand{\ma}[1]{\mathbb{#1}}
\newcommand{\mr}[1]{\mathrm{#1}}
\newcommand{\poly}{\mr{poly}}
\newcommand{\eps}{\varepsilon}

\newcommand{\EMPH}[1]{\emph{\textcolor{red!50!black}{#1}}}

\newcounter{ctr}
\loop
 \stepcounter{ctr}
 \expandafter\edef\csname c\Alph{ctr}\endcsname{\noexpand\mathcal{\Alph{ctr}}}
\ifnum\thectr<26
\repeat

\newcounter{btr}
\loop
 \stepcounter{btr}
 \expandafter\edef\csname b\Alph{btr}\endcsname{\noexpand\mathbb{\Alph{btr}}}
\ifnum\thectr<26
\repeat

\begin{document}

\maketitle
\begin{abstract}
    For $p,q\geq 2$ the $\{p,q\}$-tiling graph is the (finite or infinite) planar graph $T_{p,q}$ where all faces are cycles of length $p$ and all vertices have degree $q$. We give algorithms for the problem of recognizing subgraphs and induced subgraphs of these tiling graphs, as follows.
    \begin{itemize}
        \item For $\frac1p + \frac1q > \frac12$, these graphs correspond to the regular tilings of the sphere. These graphs are finite, thus recognizing their (induced) subgraphs can be done in constant time.
        \item For $\frac1p + \frac1q = \frac12$, these graphs correspond to the regular tilings of the Euclidean plane. For the Euclidean square grid $T_{4,4}$ Bhatt and Cosmadakis (IPL 1987) showed that recognizing subgraphs is NP-hard, even if the input graph is a tree.
        We show that a simple divide-and conquer algorithm achieves a subexponential running time in all Euclidean tilings, and we observe that there is an almost matching lower bound in $T_{4,4}$ under the Exponential Time Hypothesis via known reductions.
        \item For $\frac1p + \frac1q < \frac12$, these graphs correspond to the regular tilings of the hyperbolic plane. As our main contribution, we show that deciding if an $n$-vertex graph is isomorphic to a subgraph (or an induced subgraph) of the tiling $T_{p,q}$ can be done in quasi-polynomial ($n^{O(\log n)}$) time for any fixed $q$.
    \end{itemize}
    Our results for the hyperbolic case show that it has significantly lower complexity than the Euclidean variant, and it is unlikely to be NP-hard. The Euclidean results also suggest that the problem can be maximally hard even if the graph in question is a tree. Consequently, the known treewidth bounds for subgraphs of hyperbolic tilings do not lead to an efficient algorithm by themselves. Instead, we use convex hulls \emph{within the tiling graph}, which have several desirable properties in hyperbolic tilings. Our key technical insight is that planar subgraph isomorphism can be computed via a dynamic program that builds a (canonical) sphere cut decomposition of a solution subgraph's convex hull.
\end{abstract}

\section{Introduction}

\textsc{Subgraph Isomorphism} is the decision problem of whether a given \EMPH{pattern graph} $H$ is isomorphic to a subgraph of a given \EMPH{host graph} $G$. It is a fundamental NP-complete problem generalizing many other problems around finding patterns in graphs, such as \textsc{$k$-Clique, Hamiltonian Path, Hamiltonian Cycle}. With \textsc{Induced Subgraph Isomorphism} the goal is to decide if $H$ is isomorphic to an \emph{induced} subgraph of $G$: this generalizes other fundamental problems such as \textsc{Independent Set, Triangle Packing, Longest Induced Cycle}, etc. The complexity of these problems is widely studied, and today we have a good base understanding of their complexity.

The na\"ive brute force algorithm for \textsc{Subgraph Isomorphism} achieves a running time of $n^{O(n)}$ by trying all functions $f:V(H)\rightarrow V(G)$, and checking if the resulting mapping is a subgraph isomorphism (here $n=|V(G)|+|V(H)|$). This running time is asymptotically optimal~\cite{CyganFGKMPS16} under the Exponential Time Hypothesis (ETH)~\cite{ImpagliazzoP01}. However, this lower bound requires dense graphs, and the complexity of the problem changes when considering sparse graph classes. In particular, there is an algorithm with running time $2^{O(n/\log n)}$ for planar and minor-free graphs with a matching lower bound under ETH~\cite{BodlaenderNZ16}. Notice however that unlike many other NP-hard problems, this does not yield a \emph{subexponential} running time on planar graphs, i.e., this is one of the problems that do not have the so-called square root phenomenon~\cite{Marx13,CyganFKLMPPS15}: the existence of sublinear separators in planar graphs does not help reduce the running time to $2^{\widetilde O(\sqrt{n})}$. On the other hand, a randomized algorithm with running time $2^{O(\sqrt{|V(H)|}\log^2|V(H)|)}|V(G)|^{O(1)}$ has been proposed for the case when the host graph $G$ is planar (or more generally, apex-minor free) and the pattern $H$ is connected and has constant maximum degree~\cite{FominLMPPS22}.

In this paper, we are concerned with the variant of subgraph isomorphism where the host graph $G$ is a fixed infinite planar graph that is highly regular, more precisely, comes from a regular tiling. For $p,q\geq 2$ the \EMPH{$\{p,q\}$-tiling graph} is the (finite or infinite) planar graph $T_{p,q}$ where all faces are cycles of length $p$ and all vertices have degree $q$. In particular, the graph $T_{6,3}$ corresponds to the Euclidean hexagonal tiling. In general, these graphs correspond to the so-called regular tilings of the three classical $2$-dimensional geometries---spherical, Euclidean, hyperbolic---depending on whether $\frac1p+\frac1q$ is strictly greater than, equal to, or strictly less than $\frac12$, respectively. Observe that $T_{p,q}$ is finite if and only if it is spherical; these tilings and their (induced) subgraphs can be recognized in constant time. Thus, we will focus on Euclidean and hyperbolic tilings.

From this point onward, we will let $n$ denote the number of vertices in the pattern graph $H$. Since the graphs $T_{p,q}$ are vertex-, edge-, and face-transitive (the underlying tiling is so-called flag-transitive~\cite{GrunbaumShephard1977TilingsRegularPolygons}), we may also assume without loss of generality that $H$ is connected, as the connected components of a disconnected $H$ can be embedded at some large distance of each other, and due to transitivity we can freely choose the mapping of one vertex per component. We will let \textsc{Tiling Subgraph Recognition} and \textsc{Tiling Induced Subgraph Recognition} denote the resulting decision problems.
Our problem can be seen as a natural restriction of planar subgraph isomorphism to highly regular planar host graphs: indeed, one can take a ball of radius $\mathrm{diam}(H)$ (the shortest-path diameter of $H$) in the infinite tiling graph to define a finite host graph $G$. For Euclidean grids, the randomized algorithm~\cite{FominLMPPS22} can be run in a ball of radius $\mathrm{diam}(H)<|V(H)|$ and thus yields a running time $2^{O(\sqrt{|V(H)|}\log^2|V(H)|)}|V(H)|^{O(1)}=2^{O(\sqrt{n}\log^2 n)}$.

The problem of recognizing subgraphs of Euclidean tilings can be seen as a natural graph drawing problem: given a planar graph, can it be drawn on a square, triangular, or hexagonal grid so that each of its edges correspond to a single grid edge? The problem was already studied in the eighties: subgraph isomorphism is NP-complete if $G=T_{4,4}$ is the infinite Euclidean grid. Bhatt and Cosmadakis showed that this is true even if $H$ is a tree~\cite{BHATT1987263}.

The problem in hyperbolic tilings can also be motivated from the graph drawing/visualization perspective. Hyperbolic visualizations have been proposed for trees and other graphs with hierarchical structure~\cite{LampingRao1996HyperbolicBrowser,MunznerB95}, and embeddings into low-dimensional hyperbolic spaces are generally recognized as practically useful both in the machine learning~\cite{NickelK18,pmlr-v80-ganea18a,pmlr-v80-sala18a} and complex networks communities~\cite{Krioukov2010,Papadopoulos2012,Gugelmann2012}.

In hyperbolic tilings, Kisfaludi-Bak~\cite{Kisfaludi-Bak20} and independently, Kopczynski~\cite{Kopczynski21} proved that $n$-vertex subgraphs of hyperbolic tilings have treewidth $O(\log n)$, and this can be exploited algorithmically to gain quasi-polynomial algorithms. In fact, quasi-polynomial algorithms have been found in Gromov-hyperbolic planar graphs~\cite{Kisfaludi-BakML24}, and the tree-likeness of hyperbolic geometry can often yield much faster algorithms compared to Euclidean settings~\cite{Kisfaludi-Bak21,Blasius24}. On the other hand, as seen above, subgraph isomorphism is already hard for tree patterns, thus low treewidth  of the pattern alone should not lead to any improvements for subgraph isomorphism. We set out to answer the following question.

\begin{question}\label{mainq}
    Is there a quasi-polynomial ($n^{\poly(\log n))}$) algorithm for \textsc{Tiling (Induced) Subgraph Recognition}?
\end{question}

\subsection{Our contribution.} 
\subparagraph*{Hyperbolic Tilings.} Our main contribution concerns hyperbolic tilings. We answer \Cref{mainq} by giving a quasi-polynomial algorithm for any fixed $q$ as follows.

\begin{restatable}{theorem}{Maintheorem}\label{main_theorem}
    For any integers $p,q\geq 2$ with $\frac1p+\frac1q<\frac12$ the \textsc{Tiling Subgraph Recognition} and \textsc{Tiling Induced Subgraph Recognition} can be solved in time $2^{O(q+q\log \frac{n}{p+q})}\cdot n^{O(1+\log \frac{n}{p+q})}$.
\end{restatable}

In particular, the running time is $n^{O(\log n)}$ for any constant $q$ or even $q = O(\log n)$. We get a polynomial algorithm if $q$ is constant and $p=\Omega(n^\eps)$ for some constant $\eps>0$.

Due to the exponential expansion of the hyperbolic plane, a ball of radius $r$ in any hyperbolic tiling contains $2^{\Omega(r)}$ vertices, thus when applying~\cite{FominLMPPS22} in a ball of radius $\mathrm{diam}(H)\leq|V(H)|$ we get a running time that is potentially exponential in $n$. Thus, we need a different approach, which we can outline as follows.

We aim for a dynamic programming algorithm that builds a solution embedding $\phi:V(H)\rightarrow T_{p,q}$. Note that having a small separator in $H$ does not give an algorithm: indeed, after mapping the vertices of a separator $S$ to $T_{p,q}$ with size $|S|=O(\log n)$, the resulting subproblems (embed each connected component of $H-S$ with the mapping $\phi$ already fixed on $S$) are not independent: two different subgraphs can overlap. Thus, the resulting subproblems must be separated geometrically. A natural choice is to use a so-called \emph{sphere cut decomposition}~\cite{DornPBF10} of $H$: this has been successful in the hyperbolic setting before~\cite{Blasius24}. A sphere cut separator (usually called a \EMPH{noose}) of a plane graph $H$ is a closed curve drawn in the plane that intersects the planar drawing of $H$ in only $O(\log n)$ vertices. Note, however, that as a geometric curve, a sphere cut separator of a plane graph may require a description much more complex than the size of the corresponding vertex separator: indeed, in order to separate a tree embedded in the plane at a single vertex $v$, we must define a curve that avoids all branches and intersects the tree at only $v$ (see Figure~\ref{fig:long_separator}).

\begin{figure}[ht]
  \centering
    \centering
    \includegraphics{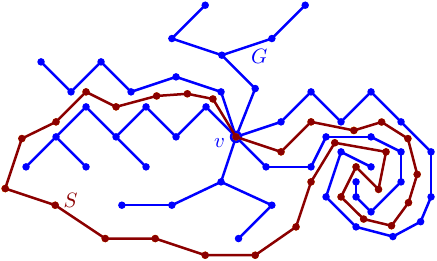} 
    \caption{A plane graph $G$ and a sphere cut separator $S$. The vertex separator corresponding to $S$ has size~$1$, however the geometric curve $S$ is much more complex.}
    \label{fig:long_separator}
\end{figure}

Consequently, rather than finding a sphere cut separator of $\phi(H)$, we use the notion of \emph{convex hulls} within graphs. The \EMPH{interval} of a pair of vertices $a,b$ in a graph $G$ is the subgraph given by all shortest paths from $a$ to $b$ in $G$. A subgraph $C$ of $G$ is \EMPH{convex} if for any $a,b\in V(C)$ we have that the interval of $a,b$ in $G$ is contained in $C$. The \EMPH{convex hull} of a vertex set $X\subset V(G)$ is the intersection of all convex subgraphs that cover $X$. It turns out that convex hulls in tiling graphs $T_{p,q}$ have several desirable properties. Recently, it was shown that the convex hull of the vertices of a connected $n$-vertex tiling graph has size $O(n)$~\cite{KisfaludiBak25}. Consequently, any fixed solution embedding $\phi(H)$ of $H$ has a sphere cut decomposition with curves that have complexity $O(\log n)$. (This requires normalizing the curves, similarly to~\cite{Blasius24}.) Our dynamic programming algorithm embeds ever larger subgraphs of $H$ (defined via a separator and its neighborhood on each side of the sphere cutting curve) into all possible sphere cuts of complexity $O(\log n)$. Unlike~\cite{Blasius24} where the goal was to get a solution for \textsc{Independent Set}, there are significant challenges remaining even after having an efficient way of enumerating all nooses: the embedding of $H$ is being built and thus the convex hull of the embedded graph is not available to the algorithm.

We note here that most of our proofs use facts from hyperbolic geometry in a black-box fashion; our arguments are mostly relying on general properties of planar graphs. As a result, our algorithm is easily understood without any prior knowledge of hyperbolic geometry.

\subparagraph*{Euclidean Tilings.} In the Euclidean setting, one can observe that rather than complicated nooses, it is sufficient to use a decomposition into rectangular regions. Compared to the randomized $2^{O(\sqrt{n}\log^2 n)}$ algorithm~\cite{FominLMPPS22}, a simple deterministic divide-and-conquer algorithm can shave a $\log n$ factor from the exponent in this special case. 

\begin{restatable}{theorem}{secondtheorem}\label{thm:second_theorem}
    For any integers $p,q\geq 2$ with $\frac1p+\frac1q=\frac12$ the \textsc{Tiling (Induced) Subgraph Recognition} problem can be solved in time $n^{O(\sqrt{n})}$.
\end{restatable}

Moreover, we observe that a composition of known reductions (notably~\cite{BHATT1987263}) imply that the running time is nearly optimal for the square grid: under the Exponential Time Hypothesis (ETH)~\cite{ImpagliazzoP01} there is no $2^{o(\sqrt{n})}$ algorithm for $\{4,4\}$-\textsc{Tiling Subgraph Recognition}.

\subsection{Further related work.}
Considered from the parameterized viewpoint, \textsc{Subgraph Isomorphism} can be solved in time $2^{O(|V(H)|)}\cdot |G|^{O(\mathrm{tw}(H))}$~\cite{Alon95}, which implies that the size of $H$ matters more than the size of $G$ for the complexity of the problem. The special case of subgraph isomorphism where both $H$ and $G$ are trees can be solved in polynomial time~\cite{MATULA197891}, however, subgraph isomorphism is NP-Hard even when both $H$ and $G$ have treewidth at most two~\cite{MatousekT92}. See~\cite{MarxP13} for an overview of several parameters and their impact on running times for \textsc{Subgraph Isomorphism}. Bodlaender et al.~\cite{BodlaenderHKKOOZ20} studied subgraph isomorphism for specific forbidden minors. Nederlof~\cite{Nederlof20} gave algorithms for detecting and counting subgraphs, achieving far-reaching generalizations of~\cite{BodlaenderNZ16}. Very recently, Lokshtanov et al.~\cite{LokshtanovPSSXZ25} gave single-exponential FPT algorithms for (a generalization of) subgraph isomorphism in the very general setting of graphs of exponential expansion; crucially, this graph class includes hyperbolic tilings, and their algorithms do not rely on planarity. In fact, the graphs obtained from 3-dimensional regular hyperbolic tilings also have exponential expansion.

The combinatorial structure of the hyperbolic tilings as well as subgraph recognition in simple transitive graphs has been explored in several works. Cannon~\cite{Cannon1984CombinatorialStructure} explores the combinatorial structure of tilings in a more general setting. Margenstern~\cite{Margenstern2006AlgorithmicTilingsPQ} considered the problem of generating spanning trees of tilings with simple formulas. The problem of recognizing cycles and closed walks in tiling graphs is also closely related to the word problem in (hyperbolic) groups~\cite{Gromov1987,Epstein92,Holt2001RealTimeWordProblem}.

\section{Preliminaries}\label{sec:prelims}

\subparagraph*{$\{p,q\}$-tilings.} A \EMPH{(polygonal) tiling} of a surface is a decomposition of the surface into pairwise interior-disjoint polygons called \emph{tiles} that cover the surface. A tiling is edge-to edge if the shared boundary of any pair of polygons is either empty, or some collection of vertices of both polygons, or a collection of full edges of both polygons. A \EMPH{regular tiling} of a surface is an edge-to-edge tiling with congruent regular polygons. The ${p,q}$-tiling $\tiling$ is the regular tiling with tiles of $p$ edges where $q$ tiles meet at each vertex. See~\cite{CoxeterMoser,ConwayGS2008,grunbaum87,adams23} for a more detailed background on tilings. We note that the 1-skeleton of $\tiling$ is a planar embedding of the graph $\til$.

We say that a graph $H\subset\til$ is \EMPH{convex}~\cite{Pelayo13}, if for any pair of vertices $u,v\in V(H)$, all shortest paths between $u$ and $v$ in $\til$ are subgraphs of $H$. A \EMPH{convex hull} of $H$, denoted $\conv(H)$, is a minimal convex graph such that $H\subseteq\conv(H)\subset\til$. 
Note that convex graphs are always connected. 

\subparagraph*{\textsc{Tiling (Induced) Subgraph Recognition}.}
For two graphs $H=(V(H),E(H))$ and $G=(V(G),E(G))$, a \EMPH{subgraph isomorphism} from $H$ to $G$ is an injective map $\phi:V(H)\to V(G)$ such that for each edge
$uv\in E(H)\ \Rightarrow\ \phi(u)\phi(v)\in E(G)$.
An \EMPH{induced subgraph isomorphism} from $H$ to $G$ is a subgraph isomorphism $\phi':V(H)\to V(G)$ such that for each edge
$uv\not\in E(H)\ \Rightarrow\ \phi'(u)\phi'(v)\not\in E(G)$.
The \textsc{Tiling Subgraph Recognition}
problem asks, given a graph $H$,
whether there exists a subgraph isomorphism from $H$ to $\til$. The \textsc{Tiling Induced Subgraph Recognition} variant asks the same with “subgraph” replaced by “induced subgraph.”

\subparagraph*{Treewidth.} A tree decomposition of a graph $H$ is a tree $\cR$ whose vertices are subsets of $V(H)$ called bags with the following properties:
(i) $\bigcup\limits_{B\in V(\cR)}B=V(H)$,
(ii) for each $uv\in E(H)$ there exists a bag $B\in V(\cR)$ such that $u,v\in B$, and
(iii) for every vertex $v\in V(H)$ the set of bags $\{B\in V(\cR)\mid v\in B\}$ induces a connected subtree of $\cR$.
The \EMPH{width} of a tree decomposition $\cR$ is $\max_{B\in V(\cR)}|B|-1$, and the \EMPH{treewidth} of a graph $H$ is the minimum width of any tree decomposition of $H$.

\subparagraph*{Poincaré disk.} We visualize the hyperbolic plane $\mathbb{H}^2$ using the \emph{Poincaré disk model}, where all points of the hyperbolic plane are assigned to points in the Euclidean unit disk, and straight hyperbolic lines appear as circular arcs perpendicular to the Euclidean unit disk. See~\cite{benedetti1992lectures,cannon1997hyperbolic,iversen1992hyperbolic} for more information. 
Figure~\ref{fig:basic_children} (i) depicts $\cT_{5,4}$ in this model.

\subparagraph*{Orientation-preserving isometry.}
An \emph{isometry} of the hyperbolic plane $\mathbb{H}^2$ is a bijection
$f : \mathbb{H}^2 \to \mathbb{H}^2$ that preserves hyperbolic distances.
An isometry is \emph{orientation-preserving} if it preserves the
clockwise/counterclockwise order of triples of points. 
In this work, we only consider orientation-preserving isometries of $\mathbb{H}^2$ that map the tiling $\tiling$ to itself.
Orientation-preserving isometries act as automorphisms on $\til$ that preserve the orientation of faces as well as the combinatorial embedding (the clockwise orientation of edges around each vertex).
See Figure~\ref{fig:candidate noose} for an illustration.

\subsection{Properties of convex hulls in $\til$.}\label{hyperbolic}

We get the following results from~\cite{KisfaludiBak25} regarding the convex hull of a graph in $\til$ which we will use to prove structural properties of sphere cut decompositions. 
Every theorem and lemma marked with a $\star$ is proved in the appendix.

\begin{lemma}[Lemma 3.1 of~\cite{KisfaludiBak25}]\label{lm:path_around_tile}
    Given a tile $\tau\in\tiling$ and two vertices $u,v\in V(\tau)$, any shortest path between $u$ and $v$ uses only vertices from $V(\tau)$.
\end{lemma}

 \begin{lemma}[Lemma 3.7 of~\cite{KisfaludiBak25}]\label{lm:kisfaudibak25_lemma3_3}
    Let $u,v\in V(\til)$ and let $S$ be the segment in $\mathbb{H}^2$ from $u$ to $v$. Any shortest path from $u$ to $v$ in $\til$ visits each vertex intersected by $S$ and for each edge $(ww')$ visited by $S$, any shortest path from $u$ to $v$ in $\til$ visits $w$ or $w'$.
\end{lemma}

\begin{lemma}[Lemma 4.3 of~\cite{KisfaludiBak25}]\label{lm:single_tile}
    If a subgraph $H$ of $\til$ is convex, then all of its bounded faces are tiles of $\til$.
\end{lemma}

\begin{restatable}[Corollary of Lemma 5.5 of~\cite{KisfaludiBak25}, $\star$]{lemma}{boundary}\label{lm:conv_O(n)}
    If $H$ is an $n$-vertex connected subgraph of $\til$, then $\conv(H)$  has $O(n)$ vertices.
\end{restatable}

\subsection{Nooses and noose hierarchies.}

\subparagraph*{Nooses and noose hierarchies.}

Let $H$ be a plane graph embedded in the unit sphere $\mathbb{S}^2$. A \EMPH{noose}~\cite{DornPBF10} of $H$ is a directed closed Jordan curve $\nu$ that intersects $H$ only on vertices of $H$  
and intersects each face of $H$ at most once. Notice that $\mathbb{S}^2\setminus\nu$ has two connected components. The \EMPH{interior} of $\nu$, denoted by $\mathrm{Int}(\nu)$, is the component on the right side of $\nu$. The \EMPH{exterior} of $\nu$, denoted by $\mathrm{Ext}(\nu)$, is the component on the left side of $\nu$. 
We denote $\reg(\nu)$ the closure of $\mathrm{Int}(\nu)$, i.e. 
$\reg(\nu)=\mathbb{S}^2\setminus \mathrm{Ext}(\nu)$.
We define the interior of the empty noose to be the entire $\mathbb{S}^2$.
We denote $T|_\nu$ the subgraph of $\til$ given the vertices and edges in $\reg(\nu)$. Note here that there may be edges induced by the vertices in $\reg(\nu)$ that are outside $\reg(\nu)$ and thus outside $T|_\nu$. 
The \EMPH{complexity} $|\nu|$ of a noose $\nu$ of $H$ is the number of vertices of $H$ intersected by $\nu$. 

A \EMPH{sphere cut decomposition}~\cite{DornPBF10} of $H$ is a rooted binary tree, where each node corresponds to a noose of $H$ such that the empty noose corresponds to the root, and every noose $\nu_P$ containing at least 3 vertices of $H$ in its region $\reg(\nu_P)$ has two child nooses $\nu_L$ and $\nu_R$, such that $\nu_L$ and $\nu_R$ are interior disjoint and $\reg(\nu_P)=\reg(\nu_L)\cup\reg(\nu_R)$, where both $\reg(\nu_L)$ and $\reg(\nu_R)$ contain strictly less vertices of $H$ than $\reg(\nu_P)$. See \Cref{fig:basic_children} (ii) for an illustration.

\begin{figure}
    \centering    \includegraphics{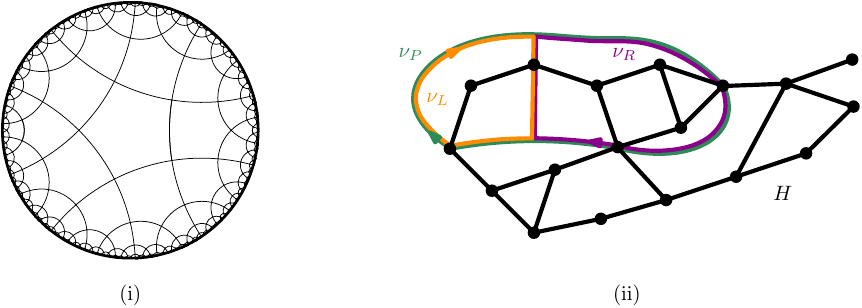}
    \caption{(i) The hyperbolic $\{5,4\}$-tiling $\cT_{5,4}$ in the Poincaré disk model. All angles are equal and all edges are (straight) line segments of the same hyperbolic length. The Poincaré disk model preserves angles, but distorts distances.
    (ii) Three nooses in a sphere cut decomposition of a graph $H$. The noose $\nu_L$ and $\nu_R$ are child nooses of the noose $\nu_P$. The region $\reg(\nu_P)$ contains $8$ vertices of $H$, and the regions $\reg(\nu_L)$ and $\reg(\nu_R)$ contain $3$ and $6$ vertices of $H$, respectively.}
    \label{fig:basic_children}
\end{figure}

To solve \textsc{Tiling (Induced) Subgraph Recognition} where $T_{p,q}$ is a hyperbolic tiling graph, we divide the problem into smaller subproblems with sphere cut decompositions. 
To achieve this, we connect nooses and sphere cut decompositions, which are defined in $\mathbb{S}^2$, with $\til$ which is defined in $\mathbb{H}^2$. 
We map a sphere cut decomposition from the sphere $\mathbb{S}^2$ to the hyperbolic plane $\mathbb{H}^2$ by puncturing one point $p$ on the sphere. Henceforth, when referring to nooses or sphere cut decompositions, we assume the mapping described above.

\section{Structural properties of connected subgraphs of \texorpdfstring{$\til$}{T}}\label{sec:Struct}

In sections \ref{sec:Struct} and \ref{sec:Algo}, we assume $\til$ to be a hyperbolic $\{p,q\}$-tiling graph, that is, $\frac1p+\frac1q<\frac12$.

\subsection{Normalized nooses.} 
For any noose $\nu$, we define a normalization 
such that, while there are an infinite number of nooses of some complexity $k$, there are only a bounded number of normalized nooses of complexity $k$. 
We first state a lemma that will be used to build the subcurve of a normalized noose through the unbounded face.
Let $\tau$ be a tile in $\tiling$. We denote $V(\tau)$ the vertices of $\til$ that represent the vertices of the polygon $\tau$. 
We call the centroid of $\tau$ the \EMPH{tile center} of $\tau$. 

\begin{restatable}[$\star$]{lemma}{bisectorray}\label{lm:angle_bisector_ray}
    Let $H$ be a graph embedded in $\til$. Given a vertex $v\in V(\conv(H))$ on the unbounded face of $\conv(H)$ and a tile $\tau\in \tiling$ on the unbounded face of $\conv(H)$ such that $v\in V(\tau)$, the bisector ray $r$ of the angle of $\tau$ at $v$ is disjoint from $\conv(H)\setminus v$.
\end{restatable}

We provide the complete description of normaized nooses in the appendix. The construction follows closely that of~\cite{Blasius24}, but differs on the unbounded face. 
A \EMPH{normalized curve}, or \EMPH{generalized polygon} in $\mathbb{H}^2$ is a plane cycle, where each vertex is either a vertex of $\til$, the center of a tile in $\til$, or a point in the unbounded face, which can be an \EMPH{ideal point}, that is, a point on the ideal boundary of $\ma H^2$. 
A \EMPH{normalized noose} $\nu$ of a convex subgraph $H$ of $\til$ is a normalized curve that 
is also a noose of $H$ and has a particular construction. Our construction through bounded faces alternates between vertices of $\til$ and tile centers.

Our construction through the unbounded face requires at most four vertices to encode, instead of the $O(n)$ vertices required for a naive construction following the same approach as for bounded faces. 
Let $\nu_\infty$ be the subcurve of a normalized noose $\nu$ through the unbounded face of $\conv(H)$.
Then $\nu_\infty$ contains the angle-bisector rays of particular tiles outside $\conv(H)$ that are connected by an ideal edge on the boundary. If the rays intersect, we instead include the corresponding \emph{prerays} (i.e., subsegments of the rays) to our normalized noose. If the rays do not intersect, we include the ideal points those rays intersect, as well as an ideal arc between those points.
By \Cref{lm:angle_bisector_ray}, these rays are disjoint from $\conv(H)$.
See Figure~\ref{fig:child_noose} for an example.

\begin{figure}[ht]
  \centering
    \includegraphics{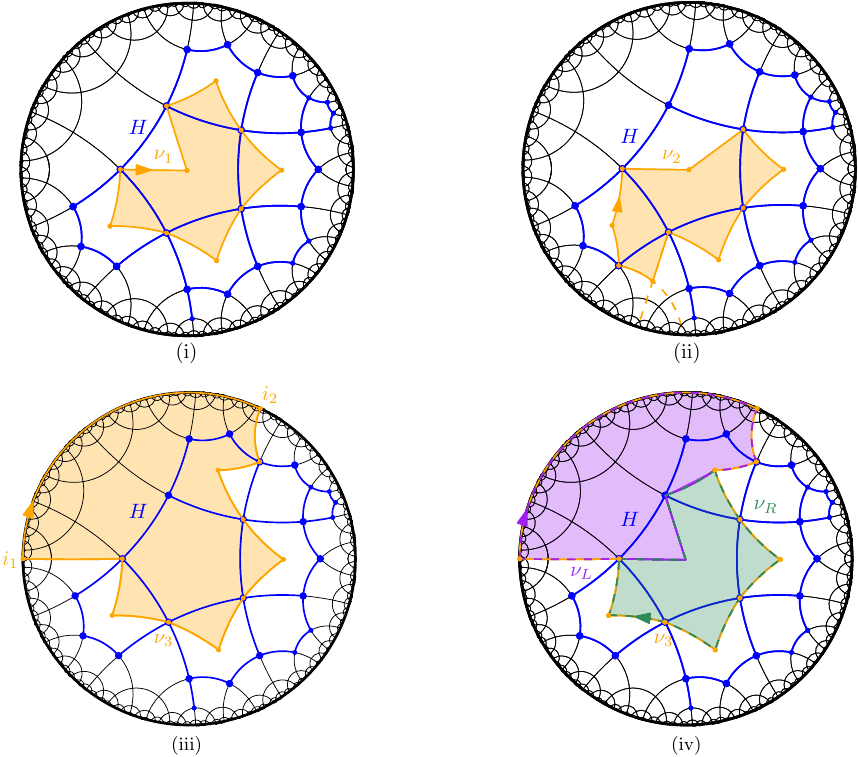}
    \caption{Each figure depicts a convex plane graph $H$ embedded in $T_{5,4}$ (in black) as well as: 
    (i) a normalized noose $\nu_1$ of $H$, such that $\nu_1=\nu_1^-$, (ii) a normalized noose $\nu_2$ of $H$ with pre-rays from two intersecting rays (the continuations of the rays are dashed), (iii) a normalized noose $\nu_3$ with two ideal vertices $i_1$ and $i_2$, (iv) two child nooses $\nu_L,\nu_R$ of $\nu_3$.}
    \label{fig:child_noose}
\end{figure}

The \EMPH{complexity} of a normalized noose $\nu$ of a convex graph $H$ is the number of vertices of $H$ that are vertices of $\nu$.
A normalized noose with complexity 0 is the \EMPH{empty noose}, and its interior is the entire plane $\mathbb{H}^2$. 
Note that according to \Cref{lm:single_tile}, each bounded face of $\conv(H)$ is a single tile, thus, 
for a noose $\nu$ of complexity $|\nu|$, the corresponding normalized noose has complexity $|\nu|$ and it follows from construction that its description contains either $2|\nu|$ or $2|\nu|+1$ vertices\footnote{Of these vertices, $|\nu|$ are vertices of $H$; either $|\nu|$ or $|\nu|-1$ are tile centers depending on whether $\nu$ visits the unbounded face; and at most $2$ are (possibly ideal) vertices on the unbounded face of $\conv(H)$.}. Since the complexity of a normalized noose, and the number of vertices in its description are within a constant factor of each other, we can use them interchangeably with big $O$ notation. 
We call a sphere cut decomposition consisting only of normalized nooses a \EMPH{normalized sphere cut decomposition}. 

We encode a noose $\nu$ such that we fix a distinguished origin vertex and encode $\nu$ as a sequence of directions.
See the appendix for the full description of the encoding of candidate nooses.
By abuse of notation, $\nu$ denotes both the candidate noose and its encoding as a sequence. 

\subsection{Compatibility.}

Next, we will define compatibility of two nooses with regards to a third noose. Compatibility is a notion that is close to the relation between parent and child nooses, however it is a characterization that works without a sphere cut decomposition. This characterization will be useful to decide whether two nooses could be children of a third noose in some sphere cut decomposition under some orientation-preserving isometries.  

We say that two closed Jordan curves $\gamma_L,\gamma_R$ in $\mathbb{S}^2$ are \EMPH{compatible} with regards to a curve $\gamma_P$, if there exist orientation-preserving isometries $\psi_L$, $\psi_R$ of $\mathbb H^2$ such that
\begin{enumerate}
    \item $\reg(\gamma_P)=\psi_L(\reg(\gamma_L))\cup\psi_R(\reg(\gamma_R))$, and
    \item there are exactly two vertices $v_1,v_2$ such that $\gamma_P\cap\psi_L(\gamma_L)\cap\psi_R(\gamma_R)= \{v_1,v_2\}$ and the symmetric difference of $\psi_L(\gamma_L)$ and $\psi_R(\gamma_R)$ is exactly $\gamma_P\setminus \{v_1,v_2\}$.
\end{enumerate}

When the two isometries are clear from context, we will simply use the notation $\gamma^*_L$ and $\gamma^*_R$ as shorthands for $\psi_L(\gamma_L)$ and $\psi_R(\gamma_R)$, respectively.

For a noose $\nu$, let $\nu^-$ be the open subcurve of $\nu$ consisting of the straight line segments of $\nu$ between a vertex of $\til$ and a tile center. So, if $\nu$ has no ideal points and no ray prefix, $\nu^-=\nu$. For the rest of this work, let $V_T(\nu)$ be the set of vertices of $V(T)$ associated with~$\nu^-$. 

We call a \EMPH{candidate noose} a normalized curve 
up to orientation preserving isometries of $\mathbb H^2$ (see Figure~\ref{fig:candidate noose}). Note that all normalized nooses of complexity $k$ can be obtained by applying some orientation-preserving isometry to some candidate noose of complexity $k$. 

\begin{figure}[ht]
  \centering
    \centering
    \includegraphics[width=0.35\textwidth]{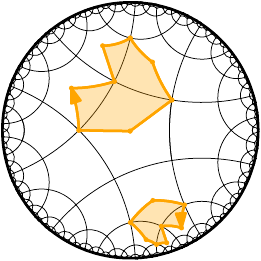} 
    \caption{The two candidate nooses depicted here are equivalent as one can be obtained from the other via an orientation-preserving isometry.}
    \label{fig:candidate noose}
\end{figure}

\subsection{Bounds on normalized nooses and sphere cut decompositions.}

Normalized sphere cut decompositions are always associated with a convex graph. However, our algorithm never has access to that graph, as it would require embedding the graph in $T_{p,q}$, which in turn requires solving \textsc{Tiling (Induced) Subgraph Recognition}. Instead, we enumerate all candidate nooses and try to construct a normalized sphere cut decomposition.

\begin{restatable}{lemma}{numbercandnooses}\label{lm:1_of3}
    There are $O((pq)^{k})$ candidate nooses of complexity $k$ or less with a fixed first vertex and first neighboring edge. Furthermore, we can enumerate all candidate nooses of complexity $k$ or less in $O((pq)^{k})$ time.
\end{restatable}

\begin{proof}
    A candidate noose $\nu$ of complexity $k$ intersects $k$ vertices of $\til$. For a noose segment from a vertex $v\in\til$ to the next vertex of $\til$, there are $q$ possible tile centers that can be the next vertex of $\nu$. After that tile center, there are $p$ vertices of $\til$ that can be the next vertex of $\nu$. So, there are $pq$ possible noose segments from $v$ to the next vertex in $V_T(\nu)$. 
    If $\nu$ has a segment through infinity, both ends of $\nu^-$ have one ray in one of $q$ possible directions, and there is only one way to connect those rays. So, there are $O((pq)^{k})$ candidate nooses of complexity $k$. We can enumerate a superset $S$ of all candidate nooses in $O((pq)^{k})$ time by following the construction described in this proof. 
    Every curve in $S$ that intersects itself during construction can be cut at the intersection to form a candidate noose of complexity less than $k$. As the construction is iterative, we will output all nooses of length less than $k$.
    Every open curve in $S$ 
    can be omitted from $S$, thus forming a set of all candidate nooses of complexity $k$ or less.
\end{proof}

We get the following bound for the treewidth of $\conv(H)$ from~\cite{KisfaludiBak25}, which will let us bound the complexity of nooses we consider.

\begin{lemma}[Theorem 1.3 of~\cite{KisfaludiBak25}]\label{lm:bound_on_treewidth}
    Let $H\subseteq\til$ be an $n$-vertex graph. The graph $\conv(H)$ has treewidth at most $\max\{12,O(\log \frac{n}{p+q})\}$.
\end{lemma} 

Now, together with \Cref{lm:bound_on_treewidth} we get the following lemma which ensures that any subgraph of $T_{p,q}$ has a sphere cut decomposition with nooses of bounded size.

\begin{restatable}{lemma}{boundednormspherecut}\label{thm:branch_width}
    Let $H$ be an $n$-vertex graph isomorphic to a subgraph of $\til$. Then there is a normalized sphere cut decomposition $\cD$ of $\conv(H)$ where each noose has complexity $O(1+\log \frac{n}{p+q})$. Moreover, $\cD$ is also a sphere cut decomposition of $H$.
\end{restatable}

\begin{proof}
    Let $\mathrm{tw}(\conv(H))$ denote the treewidth of $\conv(H)$. 
    According to \Cref{lm:conv_O(n)} \\$|V(\conv(H))|=O(n)$. So, according to \Cref{lm:bound_on_treewidth} we have that
    $\mathrm{tw}(\conv(H))=O(1+\log \frac{n}{p+q})$. According to~\cite{RobertsonS91} and~\cite[Theorem 17 of full version]{Blasius24}, there exists a sphere cut decomposition of $\conv(H)$, where every noose of $\conv(H)$ intersects at most $O(\mathrm{tw}(\conv(H)))$ vertices of $\conv(H)$. 
    Finally, for any sphere cut decomposition, there is a normalized sphere cut decomposition that preserves the complexity of nooses, and any sphere cut decomposition of $\conv(H)$ is also a sphere cut decomposition of $H$.
\end{proof}

\section{Algorithm for \textsc{Tiling (Induced) Subgraph Recognition} in hyperbolic tilings}\label{sec:Algo}

Before we introduce our algorithmic techniques, we observe that the case where $p$ is large can be handled with an easier algorithm. This is helpful as otherwise the number of candidate nooses could become very large when $p>|V(H)|$. 

\begin{restatable}{lemma}{observ}\label{obs:p,q>}
    For a graph $H$ and $p>|V(H)|$, \textsc{Tiling (Induced) Subgraph Recognition} can be decided in $O(|V(H)|)$ time.
\end{restatable}

\begin{proof}
Every cycle in $\til$ has length at least $p$, hence, if $p>|V(H)|$, $H$ can be isomorphic to a subgraph of $\til$ only if it is a tree. 
So, we can embed $H$ greedily, as when we embed a vertex $v \in V(H)$, encountering a previously embedded vertex of $H$ that is not a neighbor of $v$ in $H$ would imply a cycle of length greater than $|V(H)|$ in $H$.
Thus, any tree of degree at most $q$ is an (induced) subgraph of $T_{p,q}$ and we can decide \textsc{Tiling (Induced) Subgraph Recognition} in $O(|V(H)|)$ time. 
\end{proof}

\subsection{Boundary assignment.}

We now introduce the key components of our algorithm. One of the difficulties in designing this algorithm is to assign some subgraph of $H$ to a given noose: a planar graph $H$ has $2^{|V(H)|}$ induced subgraphs. Notice that a vertex set $B$ separating $H$ is insufficient to define a subgraph, as we need to know what connected components of $H-B$ should fall in a given noose region.

A \EMPH{boundary assignment} for a graph $H$ and candidate noose $\nu$ is a bijection $\phi: B \rightarrow V(\nu)$, where $B \subset V(H)$. We call the vertices of $B$ \EMPH{boundary vertices}. A \EMPH{boundary neighborhood constraint} is a vertex set $\widetilde B\subset N_H(B)\setminus B$, and it expresses the neighbors of $B$ that will be mapped to the interior of the noose. Later, this set together with $B$ will be used for identifying the subgraph of $H$ we want to embed into $\reg(\nu)$.
We say that a triplet of a noose, a boundary assignment, and a boundary neighborhood constraint $(\nu,\phi,\widetilde B)$ is \EMPH{valid} if:
\begin{enumerate}
    \item $\phi:B\rightarrow V(\nu)$ is a subgraph isomorphism from $H[B]$ to $T|_\nu$, and
    \item there is a partition of $V(H-B)$ into $A,A'$ such that there is no edge in $H-B$ from $A$ to $A'$, and $\widetilde B=N(B)\cap A$.
\end{enumerate}

For a boundary assignment $\phi: B \rightarrow V(\nu)$ and isometry $\psi:\nu\rightarrow\nu^*$, 
let $\phi^*=\psi \circ \phi$ be their composition, i.e., $\phi^*: B \rightarrow V(\nu^*)$.
Based on a valid triplet we are also able to find the subgraph of $H$ that we intend to embed into $\reg(\nu)$ as follows.
Let $(\nu,\phi,\widetilde B)$ be a valid triplet, and consider the subgraph induced by $B$ and the connected components of $H-B$ that intersect $\widetilde B$. We denote the resulting graph $H(\nu,\phi,\widetilde B)$. 

For a graph $H$, we want to enumerate a collection of valid triplets such that some sphere cut decomposition of $H$ is in that collection. We prove this for $p\le|V(H)|$, as otherwise we can use the algorithm in Lemma \ref{obs:p,q>}.

\begin{restatable}{lemma}{collection}\label{lm:candidate}
    Let $H$ be a given graph of degree at most $q$ with $|V(H)|\le p$ vertices. Then there is a collection $\cC$ of $2^{O(q+q\log \frac{|V(H)|}{p+q})}\cdot |V(H)|^{O(1+\log \frac{|V(H)|}{p+q})}$ valid triplets $(\nu,\phi,\widetilde B)$, where $\nu$ is a candidate noose of complexity at most $O(1+\log \frac{|V(H)|}{p+q})$, $\phi$ is a boundary assignment for $H$ and $\nu$, and $\widetilde B$ is a boundary neighborhood constraint with the following property: If $H$ is a subgraph of $\til$, then all normalized nooses and related boundary assignments corresponding to some sphere cut decomposition of $H$ are in $\cC$. Moreover, the collection  $\cC$ of triplets can be enumerated in $2^{O(q+q\log \frac{|V(H)|}{p+q})}\cdot |V(H)|^{O(1+\log \frac{|V(H)|}{p+q})}$ time. 
\end{restatable}

\begin{proof}
    According to \Cref{lm:1_of3} there are at most $O((pq)^{k})$ possible candidate nooses of complexity $k$ or less with fixed first vertex and first neighboring edge in $\til$. 
    By computing each generalized curve of complexity $k$, we will get a superset of all candidate nooses of complexity $k$.
    Each such noose has at most $k$ vertices of $\til$ and each vertex of $\til$ has $q$ neighbors in $\til$. So, a boundary neighborhood constraint associated to such a noose has size at most $kq$.
    Any boundary assignment from a subset of $H$ to $\nu$ will assign one of $|V(H)|$ vertices to $k$ different positions, so there are at most $|V(H)|^{k}$ possible boundary assignments for a noose of complexity $k$. 
    Each vertex in a boundary neighborhood constraint $\widetilde B$ has to be assigned either to the left or to the right of $\nu$, so there are $2^{kq}$ ways to assign vertices of $\widetilde B$ on a noose. According to \Cref{thm:branch_width}, there is a normalized sphere cut decomposition of $\conv(H)$, which is also a sphere cut decomposition of $H$ where each noose has complexity $O(1+\log \frac{|V(H)|}{p+q})$. 
    Hence, when $p\le |V(H)|$, the total number of triplets $(\nu,\phi,\widetilde B)$ of complexity
    $k=O(1+\log \frac{|V(H)|}{p+q})$ is 
    \[(pq)^{O(1+\log \frac{|V(H)|}{p+q})}|V(H)|^{O(1+\log \frac{|V(H)|}{p+q})}2^{q\cdot O(1+\log \frac{|V(H)|}{p+q})}=|V(H)|^{O(1+\log \frac{|V(H)|}{p+q})}\cdot 2^{O(q+q\log \frac{|V(H)|}{p+q})}\] 
\end{proof}

\subsection{Dynamic programming algorithm.}

\begin{definition}[Potential children]\label{def:potential children}
    The triplets $(\nu_L,\phi_L,\widetilde B_L)$ and $(\nu_R,\phi_R,\widetilde B_R)$ are \EMPH{potential children}  of a triplet $(\nu_P,\phi_P,\widetilde B_P)$ if
    \begin{enumerate}
        \item there exist orientation-preserving isometries $\psi_L$, $\psi_R$ such that
        \begin{itemize}
            \item the normalized curves $\nu_L$ and $\nu_R$ are compatible with regards to $\nu_P$ under the isometries $\psi_L$ and $\psi_R$, and
            \item the extension of $\phi^*_L:= \psi_L \circ \phi_L$, $\phi^*_R:=\psi_R\circ \phi_R$ and $\phi_P$ to the union of their domain is well-defined, i.e. they are equal on the shared part of their domain.
        \end{itemize}
        \item $H(\nu_P,\phi_P,\widetilde B_P)=H(\nu_L,\phi_L,\widetilde B_L)\cup H(\nu_R,\phi_R,\widetilde B_R)$, 
        \item $V(H(\nu_L,\phi_L,\tB_L))\cap V(H(\nu_R,\phi_R,\tB_R))=B_L\cap B_R$,
        \item We have $|B_L\cap B_R\cap B_P|\le 2$ and every vertex in $B_L\cup B_R\cup B_P$ appears at least in two of the sets $B_L,B_R$ and $B_P$. 
    \end{enumerate}
    
\end{definition}

Let $\overline \phi:H(\nu,\phi,\widetilde B)\rightarrow T|_\nu\subseteq\til$ be a subgraph isomorphism such that $\overline \phi|_B=\phi$ and $(\overline \phi)^{-1}(V(\nu))=B$, that is, the isomorphism $\overline\phi$ assigns only vertices of $B$ to the noose, and this assignment is the same as $\phi$.
Let $H^T(\nu,\phi,\widetilde B)$ denote the resulting embedded graph (the image of $\overline\phi$).
For the rest of this work, let $H_L$ and $H_R$ be the graphs $H(\nu_L,\phi_L,\tB_L)$ and $H(\nu_R,\phi_R,\tB_R)$, respectively, and
let $H^T_L$ and $H^T_R$ denote $H^T(\nu_L,\phi_L,\widetilde B_L)$ and $H^T(\nu_R,\phi_R,\widetilde B_R)$, respectively. 

\subparagraph*{The dynamic programming table.}

The dynamic programming table entry $A(\nu,\phi,\widetilde B)$ is \textsc{true} if and only if there exists a subgraph isomorphism $\overline \phi$ from $H(\nu,\phi,\widetilde B)$ to $T|_{\nu}\subseteq\til$ such that $\overline \phi|_B=\phi$ and $(\overline \phi)^{-1}(V(\nu))=B$, that is, the isomorphism $\overline\phi$ assigns only vertices of $B$ to the noose, and this assignment is the same as $\phi$. The entry is \textsc{false} otherwise. 

For a triplet $(\nu,\phi,\tB)$ where $|V(H(\nu,\phi,\widetilde B))|\le 2$, the triplet corresponds to a leaf noose in a sphere cut decomposition. We will discuss this after \Cref{lm:equivalence_A()_children}. 

\begin{lemma}[Recursion]\label{lm:equivalence_A()_children}
    The table entry $A(\nu_P,\phi_P,\widetilde B_P)$, where $|V(H(\nu_P,\phi_P,\wt B_P))|>2$ is \textsc{true}
    if and only if there exists a pair of triplets $(\nu_L,\phi_L,\widetilde B_L)$ and $(\nu_R,\phi_R,\widetilde B_R)$ that are potential children of $(\nu_P,\phi_P,\widetilde B_P)$ such that $A(\nu_L,\phi_L,\widetilde B_L)=A(\nu_R,\phi_R,\widetilde B_R)=$ \textsc{true}.
\end{lemma}

\begin{proof}
    The restriction $|V(H(\nu_P,\phi_P,\wt B))|>2$ guarantees that $\nu_P$ must have potential children. In this proof, we call $H_L$, $H_R$ and $H_P$ the graphs $H(\nu_L,\phi_L,\tB_L)$, $H(\nu_R,\phi_R,\tB_R)$ and $H(\nu_P,\phi_P,\tB_R)$, respectively.

    \subparagraph*{``$\Rightarrow$''} This is the straightforward direction. Let us assume that $A(\nu_P,\phi_P,\widetilde B_P)=$ \textsc{true}, so there exists an embedding of $H_P$ to $\til$. According to \Cref{thm:branch_width} there is a normalized sphere cut decomposition of ${H_P}$, where $\nu_P$ has two compatible child nooses $\nu_L$ and $\nu_R$.
    The compatibility of $\nu_L$ and $\nu_R$ with regards to $\nu_P$ follows directly from the noose hierarchy. All vertices of $H(\nu_L,\phi_L,\tB_L)$ and $H(\nu_R,\phi_R,\tB_R)$ can be assigned to $\til$ the same way as in $H^T(\nu_P,\phi_P,\widetilde B_P)$. So, $A(\nu_L,\phi_L,\widetilde B_L)=A(\nu_R,\phi_R,\widetilde B_R)=$ \textsc{true}.
    
    \subparagraph*{``$\Leftarrow$''} 
    Let us assume that $(\nu_L,\phi_L,\widetilde B_L)$ and $(\nu_R,\phi_R,\widetilde B_R)$ are potential children of $(\nu_P,\phi_P,\widetilde B_P)$, and that $A(\nu_L,\phi_L,\widetilde B_L)=A(\nu_R,\phi_R,\widetilde B_R)=$ \textsc{true}. 
    So, there exists two isomorphisms $\overline\phi_L\colon H_L\rightarrow H_L^T\subseteq T|_{\nu_L}$ and $\overline\phi_R\colon H_R\rightarrow H_R^T\subseteq T|_{\nu_R}$. 
    Let $H^T_L,H^T_R\subseteq\til$ be the images of $\bp_L$ and $\bp_R$, respectively, that is, $\bp(H_L)=H^T_L$ and $\bp(H_R)=H^T_R$. We denote $\bp^*_L$ and $\bp^*_R$ the boundary assignments as in \Cref{def:potential children}.
    Now, we define \[
        \overline\phi_P(v)=\left\{
        \begin{aligned}
        \overline\phi^*_L(v), & \text{ if } v\in V(H_L) \\
        \overline\phi^*_R(v), & \text{ if } v\in V(H_R) 
        \end{aligned}
        \right.
    \]

    We claim that the function $\overline\phi_P$ is well-defined, that is, $\overline\phi^*_L(v)=\overline\phi^*_R(v)$ for all $v\in \dom(\overline\phi^*_L)\cap\dom(\overline\phi^*_R)$. To see this, let $v\in \dom(\overline\phi^*_L)\cap\dom(\overline\phi^*_R)= V(H_L)\cap V(H_R)=B_L\cap B_R$, where the last equality follows from the fact that $(\nu_L,\phi_L,\widetilde B_L)$ and $(\nu_R,\phi_R,\widetilde B_R)$ are potential children of $(\nu_P,\phi_P,\widetilde B_P)$. By the definition of the dynamic programming table we have $\overline\phi^*_L(v)=\phi_L(v)$ and $\overline\phi^*_R(v)=\phi_R(v)$.
    So, as $\phi^*_L$ and $\phi^*_R$ are well-defined on their shared part of their domain, it follows that $\overline\phi^*_L(v)=\phi_L^*(v)=\phi^*_R(v)=\overline\phi^*_R(v)$, as required.

    Next, we will prove in the following two claims that $\overline\phi_P$ is an isomorphism.
    
    \begin{claim*}\label{lm:phi_P_bijection}
        The function $\overline \phi_P$ is injective.
    \end{claim*}
    \begin{claimproof}
        Let $u\not=v$ be two vertices in $V(H_L)\cup V(H_R)$. We assume for the sake of contradiction that $\bp_P(u)=\bp_P(v)$. If $u,v\in V(H_L)$, then $\bp_P(u)=\bp^*_L(u)$ and $\bp_P(v)=\bp^*_L(v)$, however $\bp^*_L$ is an isomorphism, so $\ba\phi_P(u)\not=\ba\phi_P(v)$. A similar argument is valid if $u,v\in V(H_R)$.

        If $u\in V(H_L)\setminus V(H_R)$ and $v\in V(H_R)\setminus V(H_L)$, then according to Condition 3 of \Cref{def:potential children}, $u,v\not\in B_L\cap B_R$. 
        So, $\bp_P(u)=\bp^*_L(u)$ and $\ba\phi_P(v)=\ba\phi^*_R(v)$. The function $\ba\phi^*_L$ assigns vertices of $H_L\setminus B_R$ to $ T|_{\nu^*_L}\setminus\nu^*_R$, so $\bp^*_L(u)\in T|_{\nu^*_L}\setminus\nu^*_R$. The curves $\nu_L$ and $\nu_R$ are compatible with regards to $\nu_P$, so $ T|_{\nu^*_L}\setminus\nu^*_R= T|_{\nu^*_L}\setminus  T|_{\nu^*_R}$. Thus, $\bp_P(u)\in  T|_{\nu^*_L}\setminus  T|_{\nu^*_R}$.
        Similarly, $\bp_P(v)\in  T|_{\nu^*_R}\setminus  T|_{\nu^*_L}$, so $\bp(u)\not=\bp(v)$, which forms a contradiction. 
    \end{claimproof}

    \begin{claim*}
    \label{lm:edge_also_in_phi_P}
        The function $\overline \phi_P$ is an isomorphism, that is, $uv\in E(H_P)$ if and only if $\bp_P(u)\bp_P(v)\in E(H_P^T)$.
    \end{claim*}
    \begin{claimproof}
        If $u,v\in V(H_L)$ or $u,v\in V(H_R)$, the claim holds, as $\bp^*_L$ and $\bp^*_R$ are isomorphisms.

        Now, let $u\in V(H_L)\setminus V(H_R)$ and $v\in V(H_R)\setminus V(H_L)$. Observe that $B_L\cap B_R$ separates $V(H_L)$ from $V(H_R)$ inside $V(H_P)$ thus $uv\not\in E(H_P)$. Now, we will prove that $\bp_P(u)\bp_P(v)\not\in E(H_P^T)$.
        As in the proof of the previous claim, $\bp_P(u)=\bp_L^*(u)\in T|_{\nu^*_L}\setminus T|_{\nu^*_R}$ and $\bp_P(v)=\bp_R^*(v)\in T|_{\nu^*_R}\setminus T|_{\nu^*_L}$. Notice that $\nu^*_L\cap \nu^*_R$ separates $\reg(\nu_P)$ and contains neither $u$ nor $v$, thus $uv\not\in E(T|_{\nu})$. 
    \end{claimproof}

    Next, we will prove that $\bp_P|_{B_P}=\phi_P$ and $(\bp_P)^{-1}(V_T(\nu_P))=B_P$.
    First, we prove that $\bp_P|_{B_P}=\phi_P$.
    Let $v\in B_P$. If $v\in B_L\subseteq V(H_L)$, then $\bp_P|_{B_P}(v)=\bp^*_L|_{B_P\cap B_L}(v)=\phi^*_L|_{B_P}(v)=\phi_P|_{B_L}(v)$. Otherwise, $v\not\in B_L$, so by $B_P\subset B_L\cup B_R$, we have $v\in B_R$, and the symmetric argument applies.  So, $\bp_P|_{B_P}=\phi_P$.
       
    Now, we prove that $(\bp_P)^{-1}(V_T(\nu_P))=B_P$.
    Let $w\in V_T(\nu_P)$. If $w\in\nu_P\cap\nu^*_L$, then
    \[w\in(\overline\phi_P)^{-1}(V_T(\nu_P\cap\nu^*_L))=(\bp_L^*)^{-1}(V_T(\nu_P\cap\nu^*_L))=(\phi_L^*)^{-1}(V_T(\nu_P))=B_P\cap B_L.\] 
    If $w\in V_T(\nu_P\setminus\nu^*_L)$, then $w\in\nu_P\cap \nu^*_R$ as $\nu_P\subseteq \nu^*_L \cup \nu^*_R$. Thus, by symmetry $w\in B_P\setminus B_L$, as required.    
 Finally, there exists a subgraph isomorphism $\overline \phi_P$ from $H(\nu_P,\phi_P,\widetilde B_P)$ to $\reg(\nu_P)\subseteq\til$ such that $\overline \phi_P|_B=\phi_P$ and $(\overline \phi_P)^{-1}(V(\nu_P))=B_P$. 
    Thus, $A(\nu_P,\phi_P,\widetilde B_P)=$ \textsc{true}.
\end{proof}

\subparagraph*{Leaf nooses.} 
As described in the definition of the dynamic programming table, each triplet $(\nu,\phi,\tB)$ where $|H(\nu,\phi,\tB)|\le 2$ corresponds to a leaf noose in a sphere cut decomposition. In this case, it can be decided in constant time whether there exists a subgraph isomorphism $\overline \phi$ from $H(\nu,\phi,\widetilde B)$ to $T|_{\nu}\subseteq\til$ such that $\overline \phi|_B=\phi$ and $(\overline \phi)^{-1}(V(\nu))=B$.

\Maintheorem*

\begin{proof}
    Let $H$ be an $n$-vertex graph. If $H$ has a vertex of degree at least $q+1$, then we reject; this can be done in $O(n^2)$ time.
    According to \Cref{lm:1_of3} and \Cref{lm:candidate} there are $(pq)^{O(1+\log{\frac{n}{p+q}})}$ valid triplets $(\nu,\phi,\tB)$ of size $O(1+\log{\frac{n}{p+q}})$ and we can enumerate all of them in $2^{O(q+q\log \frac{n}{p+q})}\cdot n^{O(1+\log \frac{n}{p+q})}$ time; let $\cC$ denote the set of enumerated triplets. 
    Our next goal is to sort all valid triplets in $\cC$ in increasing order of $|V(H(\nu,\phi,\tB))|$. 
    To determine the sizes of the graphs $H(\nu,\phi,\tB)$, we must first construct them, which can be done as follows. 
    We introduce an auxiliary source vertex $s$ (not in $H$) and connect $s$ to every vertex of $B\subseteq H$, that is, every vertex that $\phi$ assigns to $\nu$.
    Then, we perform a breadth-first search in $H\cup s$ starting from $s$, such that for vertices of $B$ we consider only neighbors in $B\cup\widetilde B$. As only vertices of $B$ can be neighbors of vertices in $H\setminus H(\nu,\phi,\tB)$, 
    this construction results in $H(\nu,\phi,\tB)$.
    The breadth first search takes $O(|V(H(\nu,\phi,\tB))|)$ time, which is at most $O(n)$. Finally, we sort all valid triplets according to their size $|V(H(\nu,\phi,\tB))|$ in time $O(|\cC|\cdot\log|\cC|)$.  
    
    The dynamic programming algorithm will process each triplet in $\mathcal{C}$ in their given order, thus, for any triplet $(\nu,\phi,\tB)$, all its potential children are processed before it. According to \Cref{lm:equivalence_A()_children}, the dynamic programming table entry of a triplet $(\nu,\phi,\tB)$, where $|H(\nu,\phi,\tB)|> 2$ can be decided by finding every pair of potential children. Since the potential children of a triplet are processed before the triplet itself, it is sufficient to process every triplet only once.

    \subparagraph*{Computing potential children.} In order to run a dynamic programming algorithm based on the recursion given in \Cref{lm:equivalence_A()_children} we need to decide algorithmically whether the triplets $(\nu_L,\phi_L,\tB_L)$ and $(\nu_R,\phi_R,\tB_R)$ are potential children of a valid triplet $(\nu_P,\phi_P,\tB_P)$.
    In Condition 2 of \Cref{def:potential children}, we compare with each other three graphs with at most $|V(H_P)|$ vertices and at most $|E(H_P)|$ edges. In conditions 3 and 4, we compare a constant number of sets of size at most $|V(H_P)|$ with each other.
    In our encoding of candidate nooses, we can fix the origin such that each vertex of a noose has distance at most $|V(H)|^{O(1)}$ from the origin.
    So, conditions 2--4 can be computed in polynomial time. 

    Next, we will show how to verify Condition 1 of \Cref{def:potential children} in polynomial time. We need to compute whether $\nu_L$ and $\nu_R$ are compatible with regards to $\nu_P$ under some isometries $\psi_L$ and $\psi_R$ such that the extension of $\psi_L\circ\phi_L,\psi_R\circ\phi_R$ and $\phi_P$ to the union of their domain is well-defined. 
    We will not compute $\psi_L$ and $\psi_R$ explicitly; instead, we will compute $\nu^*_L=\psi_L(\nu_L)$ and $\nu^*_R=\psi_R(\nu_R)$ directly.
    In any sphere cut decomposition of a graph $H$, the empty noose corresponds to the root and every other noose has complexity at least 1. If $\nu_P$ does not correspond to the empty noose, $\phi_P$ assigns at least one vertex $v$ of $H$ to $\nu_P$. Every vertex that is assigned to $\nu_P$ by $\phi_P$ is also assigned to at least one of $\nu_L$ and $\nu_R$ by $\phi_L$ or $\phi_R$, respectively. So we have $\nu_i$, where $i\in\{L,R\}$, that is assigned the vertex $v$ by $\phi_i$. This fixes a single way to fit $\nu_i$ and $\nu_P$ together, so there is only one way to define $\nu^*_i$, if $\nu^*_i$ exists, and we can compute that in polynomial time. Also, since $\nu_j$, where $j\in\{L,R\}$ and $j\not=i$ does not correspond to the empty noose, $\nu_j$ is assigned at least one vertex $w$ of $H$ by $\phi_j$. Again, at least one of $\nu_P$ and $\nu_i$ is assigned $w$ by $\phi_P$ or $\phi_i$, respectively. This results in a single way to fit $\nu_L,\nu_R$ and $\nu_P$ together. So, there is also only one way to define $\nu^*_j$, if $\nu^*_j$ exists.
    so, it can be decided in polynomial time whether $\nu^*_L$ and $\nu^*_R$ exist. Thus, it can also be decided whether
    $\nu_L$ and $\nu_R$ are compatible with regards to $\nu_P$ under $\psi_L$ and $\psi_R$, and whether the extension of $\psi_L\circ\phi_L,\psi_R\circ\phi_R$ and $\phi_P$ to the union of their domain is well-defined, as in all of those it is sufficient to compare 3 sets of at most $|V(H_P)|$ vertices with each other.

    If $\nu_P$ is the empty noose, that is, $\phi_P$ assigns no vertices to $\nu_P$, then if $\nu_L$ and $\nu_R$ are potential children of $\nu_P$, it must hold that $B_L=B_R$ and $|B_L|\ge1$, which can be decided in polynomial time. It must also hold that there is a common vertex assigned to $\nu_L$ and $\nu_R$ by $\phi_L$ and $\phi_R$, respectively. So, there is only one way to fit $\nu_L$ and $\nu_R$ together. Thus compatibility between $\nu_L$ and $\nu_R$, as well as the nooses $\nu^*_L$ and $\nu^*_R$, if they exist, can be computed in polynomial time. The nooses $\nu_L$ and $\nu_R$ are compatible with $\nu_P$ if $\reg(\nu^*_L)\cup\reg(\nu^*_R)$ is the entire plane. This can be decided in polynomial time by following the boundaries of $\nu^*_L$ and $\nu^*_R$.
    Finally, according to \Cref{lm:equivalence_A()_children}, if and only if the table entry $A(\nu_\emptyset,\phi_\emptyset,\tB_\emptyset)$, where $\nu_\emptyset$ is the empty noose, is \textsc{true}, then $H$ is isomorphic to a subgraph of $\til$.

    \subparagraph*{Induced variant of the problem.} If we are solving the induced variant of \textsc{Tiling Subgraph Recognition}, we need the additional restriction that we cannot embed two non-neighboring vertices of $H$ to neighboring vertices of $\til$ when computing leaf triplets and when computing whether two triplets are potential children of a third triplet.

    \subparagraph*{Total running time.} Recall that we enumerated and sorted all $(pq)^{O(1+\log{\frac{n}{p+q}})}$ valid triplets in time $2^{O(q+q\log \frac{n}{p+q})}\cdot n^{O(1+\log \frac{n}{p+q})}$. Then, we processed the $(pq)^{O(1+\log{\frac{n}{p+q}})}$ valid triplets in order. Triplets with a graph of complexity at most $2$ are processed in constant time, and other triplets are processed by matching pairs of nooses from a collection of at most $(pq)^{O(1+\log{\frac{n}{p+q}})}$ nooses, where each pair takes $n^{O(1)}$ time. Finally, we return the table entry of the valid triplet containing the empty noose. This takes a total time of 
    \[
    2^{O(q+q\log \frac{n}{p+q})}\cdot n^{O(1+\log \frac{n}{p+q})} 
    + (pq)^{O(1+\log{\frac{n}{p+q}})}\cdot n^{O(1)}
    = 2^{O(q+q\log \frac{n}{p+q})}\cdot n^{O(1+\log \frac{n}{p+q})}.\qedhere
    \]
\end{proof}

\section{Algorithm and lower bound for \textsc{Tiling (Induced) Subgraph Recognition} in Euclidean tilings}\label{sec:Euclidean}

In this section, $\til$ always represents a Euclidean $\{p,q\}$-tiling graph, that is, $\frac1p+\frac1q=\frac1 2$.

\subsection{\{p,q\}-\textsc{(Induced) Tiling Subgraph Recognition}.}

We use a simple divide \& conquer scheme based on the following separator theorem.

\begin{theorem}[Corollary of~\cite{CarmiCKKORRSS20}]\label{thm:separator44}
    Let $H$ be an $n$-vertex subgraph of the Euclidean grid $T_{4,4}$. Then there exists a line $\ell$ that is axis-aligned such that $S:=V(H)\cap \ell$ is vertex separator of size $|S|=O(\sqrt{n})$ such that both half-planes of $\ell$ contain at most $\frac45 n$ vertices of $H-S$.
\end{theorem}

\begin{restatable}{corollary}{corol}\label{cor:separator36}
    Let $H$ be an $n$-vertex subgraph of $T_{3,6}$ or $T_{6,3}$ where the edge directions are among the unit vectors $e_1=(1,0),e_2=(\frac12,\frac{\sqrt3}2),e_3=(\frac{-1}2,\frac{\sqrt3}2)$. Then there exists a line $\ell$ that is parallel to one of $e_1,e_2$ such that $S:=V(H)\cap \ell$ is separator of size $|S|=O(\sqrt{n})$ such that both half-planes of $\ell$ contain at most $\frac45 n$ vertices of $H-S$.
\end{restatable}

\begin{proof}
    There exists a linear transformation that maps $e_1$ and $e_2$ to the orthogonal basis $(0,1)$ and $(1,0)$. That transformation maps vertices and edges parallel to $e_1$ or $e_2$ of $T_{3,6}$ and $T_{6,3}$ to vertices and edges of $T_{4,4}$, respectively. Edges of $T_{3,6}$ and $T_{6,3}$ parallel to $e_3$ correspond to diagonals of the tiles of $T_{4,4}$ parallel to $(-1,1)$. So, vertices are mapped to vertices of $T_{4,4}$, moreover, vertical and horizontal grid lines only intersect vertices and no edges; thus forming a vertex separator on the vertices of $T_{3,6}$ and $T_{6,3}$. 
\end{proof}

For this section, if $\til$ is the tiling graph $T_{3,6}$ or $T_{6,3}$, we assume the geometric transformation from the proof of \Cref{cor:separator36}, so that vertices are a subset of the integer grid~$\mathbb{Z}^2$.

If an $n$-vertex graph $H$ is isomorphic to a subgraph of $\til$, it is also isomorphic to a subgraph of $\til$ bounded by a square boundary $\rho$ where each side is axis-aligned and has length $2n$.
This follows from the observation that any two vertices in $H$ have distance at most $n-1$. 
We denote $\reg(\rho)$ the closure of the interior of $\rho$.
Conversely, if $H$ is isomorphic to a subgraph of $\til\cap \reg(\rho)$, then $H$ is isomorphic to a subgraph of $\til$. 
 Thus, it is sufficient to solve \textsc{Tiling (Induced) Subgraph Recognition} from $H$ to the bounded region $\reg(\rho)$.
We will solve this problem using a technique similar to the solution for \textsc{Tiling (Induced) Subgraph Recognition} for hyperbolic tiling graphs with some distinctions, namely, the desired running time allows the usage of a simpler divide \& conquer algorithm instead of the dynamic programming algorithm. Furthermore, the size of the convex hull of an $n$-vertex subgraph of $T_{4,4}$ is $O(n^2)$, unlike hyperbolic $\{p,q\}$-tiling graphs where it was $O(n)$.
So, instead of dividing our problem with nooses, we work with a rectangle-shaped bounded regions that we divide with axis-aligned lines. 

We say that a rectangle or line is \EMPH{grid-aligned} if its sides are axis-aligned and its vertices are grid points. 
We use $V(\rho)$ to refer to the set of grid points that fall on the boundary $\rho$.
Observe that the boundary $\rho$ of a grid-aligned rectangle of $\til$ and an edge $e$ of $\til$ are either disjoint, intersect in one or two endpoints of $e$, or $e\subset \rho$. Consequently, the boundary of a grid-aligned rectangle induces a vertex separator in $\til$.

\subparagraph*{Rectangle-restricted problem.}

We use the same notion of valid triplets $(\rho,\phi,\wt B)$ as in \Cref{sec:Algo}, with the difference that instead of using nooses, we use grid-aligned rectangle boundaries $\rho$. 
Let $H$ be a graph with $n$ vertices, and let $(\rho,\phi,\wt B)$ be a triplet consisting of a grid-aligned rectangle boundary, a boundary assignment, and a boundary neighborhood constraint.
We want to decide whether $H(\rho,\phi,\wt B)\subseteq H$ can be embedded to $\til$ inside $\reg(\rho)$ satisfying~$\phi$ and the restriction given by the boundary neighbors $\wt B$.

In this setting, a triplet where $\reg(\rho)$ is a single grid tile represents a leaf triplet and can be solved in constant time by exhaustive enumeration. 
For a triplet $(\rho,\phi,\wt B)$ where $\reg(\rho)$ is larger than a single grid tile, the problem can be split into two child triplets $(\rho_L,\phi_L,\wt B_L)$ and $(\rho_R,\phi_R,\wt B_R)$ by splitting $\reg(\rho)$ with a grid-aligned line $S$. Consequently, both $\rho_L$ and $\rho_R$ are grid-aligned rectangle boundaries such that one side is $S$ and the three other sides come from $\rho$. Similarly, $\phi_L|_{\rho\cap\rho_L}=\phi|_{\rho\cap\rho_L}$ and we assign new vertices of $H$ to $\phi_L|_{S}$. The same is true for $\phi_R$.
The problem can be solved similarly as in \Cref{sec:Algo}, except that instead of computing valid triplets from smallest to largest, we start with the root triplet where $\rho$ is a $2n\times 2n$ grid-aligned square boundary, $\phi$ is the empty function and $\wt B=\emptyset$. Then we divide the problem recursively into all possible pairs of child triplets with all possible grid-aligned lines.

\secondtheorem*

\begin{proof}
Let $H_0$ be an $N$-vertex graph and let $(\rho_0,\phi_0,\wt B_0)$ be a triplet consisting of a $2N\times 2N$-grid-aligned square $\rho_0$, the empty function $\phi_0$, and the boundary neighborhood constraint $\wt B_0=\emptyset$. We use the divide \& conquer algorithm described above to decide whether $H$ can be embedded inside~$\rho_0$.
According to  \cref{thm:separator44} and  \cref{cor:separator36}, for any embedding of a graph $H$ to a region $\reg(\rho)$, there exists a grid-aligned line $S$ that divides $\rho$ into two grid-aligned rectangles, such that $S$ is a vertex separator of $H$, there are at most $\frac45|V(H)|$ vertices of $H$ on both sides and there are at most $O(\sqrt{|V(H)|})$ vertices of $H$ on $S$.

For each rectangle-restricted problem in our algorithm where we have a graph $H\subseteq H_0$ with size $|V(H)|=n\le N$, a grid-aligned rectangle $\rho$ such that $\reg(\rho)\subseteq\reg(\rho_0)$, a set of boundary vertices $B\subseteq H$, a boundary assignment $\phi:B\to\rho$, and a boundary neighborhood constraint $\wt B\subseteq H$. We divide $(\rho,\phi,\wt B)$ into all possible pairs of child triplets $(\rho_L,\phi_L,\wt B_L)$ and $(\rho_R,\phi_R,\wt B_R)$.
There are $O(N)$ grid-aligned lines $S$ that divide $\reg(\rho)$ into two regions $\rho_L$ and $\rho_R$. For each such line, we compute all boundary assignments $\phi_S:V(H\setminus B)\to V(S\cap \til)$ of size $O(\sqrt{n})$, which will be part of both $\phi_L$ and $\phi_R$. There is a total of $n^{O(\sqrt{n})}$ subsets $B_S$ of $V(H\setminus B)$ of size $O(\sqrt{n})$. Since $S$ has size at most $2N$, there are $N^{O(\sqrt{n})}$ subsets of $V(\til\cap S)$ of size $O(\sqrt{n})$.
Thus, there is a total of $n^{O(\sqrt{n})}\cdot N^{O(\sqrt{n})}=N^{O(\sqrt{n})}$ possible boundary assignments $\phi_S:B_S\to S\cap\til$. For each vertex in a subset $B_S$, that vertex has at most 4 neighbors, so there is a total of $O(\sqrt{n})$ neighbors of vertices of $B_S$. For each neighbor, we assign it to either side of $S$ (either to $\wt B_L$ or to $\wt B_R$), which makes a total of $2^{O(\sqrt{n})}$ assignments of neighbors. So, one subproblem of size $n$ has a number of recursive calls equal to $O(N)\cdot n^{O(\sqrt{n})}\cdot N^{O(\sqrt{n})}\cdot 2^{O(\sqrt{n})}=N^{O(\sqrt{n})}$, each with a subproblem of size at most $\frac{4}{5}n$.
So, we solve a subproblem $T(n)$ in time ($C$ is a constant):

\begin{multline*}
T(n)=N^{O(\sqrt{n})}\cdot T\left(\frac45 n\right)
=
N^{C\cdot\sqrt{n}}\cdot N^{C\cdot\sqrt{\frac{4}{5}n}}\cdot N^{C\cdot\sqrt{\left(\frac{4}{5}\right)^2n}}\cdot ...\cdot N^{C\cdot\sqrt{\left(\frac{4}{5}\right)^in}}\\
<
\exp(\log N \cdot \big[C\cdot\sqrt{n}\cdot\sum_{i=0}^\infty\left(\sqrt{\frac45}\right)^i  \big])=N^{O(\sqrt{n})}.
\end{multline*}

Thus the original instance can be solved in time $T(N)=N^{O(\sqrt{N})}$.
\end{proof}

\subsection{Lower bound.}

We prove a lower bound for \textsc{Tiling (Induced) Subgraph Recognition}. 
Our reduction includes three variations of \textsc{3-CNFSAT}.
\begin{definition}[\textsc{3-CNFSAT}]
    An instance of \textsc{3-CNFSAT} (or \textsc{3-SAT}) is a Boolean formula $\Phi=\bigwedge_{i=1}^{m} (\ell_{i1}\vee\ell_{i2}\vee\ell_{i3})$, where each $\ell_{ij}$ is a literal, that is,
    either $x_k$ or $\neg x_k$ for some boolean variable $x_k\in\{x_1,\dots,x_n\}$. 
    The question is whether there exists a truth assignment $\alpha:\{x_1,\dots,x_n\}\to\{\textsc{true},\textsc{false}\}$ that satisfies every clause of $\Phi$.
\end{definition}

We define the following two variations of \textsc{3-SAT}.
\begin{itemize}
    \item \textsc{(3,3)-SAT}: Given a formula $\Phi$ of \textsc{3-SAT}, where each clause has at most 3 literals and each variable appears at most 3 times, is there an assignment $\alpha$ that satisfies every clause of $\Phi$.
    \item \textsc{Not-All-Equal 3CNFSAT}, or \textsc{NAE3SAT}: Given a formula $\Phi$ of \textsc{3-CNFSAT}, is there an assignment $\alpha$ that satisfies every clause of $\Phi$, such that every clause has at least one false literal.
\end{itemize}

We will also use the following observation, see~\cite{BergBKMZ20}.

\begin{observation}[\cite{BergBKMZ20}]\label{obs:33size}
    If $\phi$ is a \textsc{(3,3)-SAT} formula on $n$ variables, then it has $O(n)$ clauses and has total size $O(n)$.
\end{observation}

\begin{restatable}{lemma}{eth}\label{lm:reduction}
    There is a polynomial time reduction that, given an instance $\Phi$ of \textsc{(3,3)-SAT} on $n$ variables, produces an equivalent instance of $\{4,4\}$-\textsc{Tiling (Induced) Subgraph Recognition} on $O(n^2)$ vertices. 
\end{restatable}

\begin{proof}
    There is a polynomial time reduction from \textsc{3-SAT} to \textsc{NAE3SAT} presented in~\cite{williams2007algorithms} where each original clause and variable corresponds to constantly many new clauses and variables.
    We can do a similar reduction from \textsc{(3,3)-SAT} to \textsc{NAE3SAT} as follows: 
    \begin{itemize}
        \item For a clause $(\ell_1\vee\ell_2\vee\ell_3)$ of three variables of \textsc{(3,3)-SAT}, we construct the pair of clauses $(\ell_1\vee\ell_2\vee v_{c_i}),(\neg v_{c_i}\vee\ell_3\vee w)$ as in the reduction of~\cite{williams2007algorithms}, where $v_{c_i}$ and $w$ are newly introduced variables.
        \item For a clause $(\ell_1\vee\ell_2)$ of two variables of \textsc{(3,3)-SAT}, we construct the pair of clauses $(\ell_1\vee\ell_2\vee v_{c_i}),(\neg v_{c_i}\vee w'\vee w)$, where $w'$ is a variable analogous to $w$.
        \item For a clause $(\ell_1)$ of one variable of \textsc{(3,3)-SAT}, we construct the pair of clauses $(\ell_1\vee w''\vee v_{c_i}),(\neg v_{c_i}\vee w'\vee w)$, where $w''$ is a variable analogous to $w$ and $w''$.
    \end{itemize}
    Otherwise, our reduction is similar to the reduction from \textsc{3-SAT} to \textsc{NAE3SAT}. This creates an equivalent instance of \textsc{NAE3SAT} with $O(n)$ varaibles and $O(n)$ clauses.

    According to~\cite{BHATT1987263}, there is a polynomial  time reduction from
    \textsc{NAE3SAT}
    to $\{4,4\}$-\textsc{Tiling Subgraph Recognition}, where a formula of \textsc{NAE3SAT} with $n$ variables and $m$ clauses becomes a tree of size $O(mn)$.
\end{proof}

Finally, we prove a lower bound for \textsc{Tiling Subgraph Recognition}
under the Exponential Time Hypothesis (ETH)~\cite{ImpagliazzoP01}.

\begin{theorem}
    Unless ETH fails, \textsc{Tiling Subgraph Recognition} admits no algorithm with running time $2^{o(\sqrt{n})}$, where $n$ is the number of vertices in the pattern graph.
\end{theorem}

\begin{proof}
    Suppose that there exists an algorithm with running time $2^{o(\sqrt{n})}$ for \textsc{Tiling Subgraph Recognition}. According to \Cref{lm:reduction}, there is a polynomial time reduction of from
    \textsc{(3,3)-SAT} to $\{4,4\}$-\textsc{Tiling (Induced) Subgraph Recognition} where for a formula of size $n$ we obtain an equivalent instance of $\{4,4\}$-\textsc{Tiling (Induced) Subgraph Recognition} of size $N=O(n^2)$. Thus, composing this reduction with a hypothetical $2^{o(\sqrt{n})}$ algorithm yields an algorithm for \textsc{(3,3)-SAT} with running time
    $\poly(n)+ 2^{o(\sqrt{O(N)})}=2^{o(n)}$. However, according to~\cite[Proposition 3.1]{BergBKMZ20}  , such an algorithm does not exist unless ETH fails, which forms a contradiction.
\end{proof}

\section{Conclusion}

In this paper we gave a quasi-polynomial algorithm for recognizing subgraphs of hyperbolic tilings, and a nearly ETH-tight algorithm for recognizing subgraphs of Euclidean tilings. Several important open questions remain for future work.
\begin{itemize}
\item Is there an $n^{O(\log n)}$ algorithm for recognizing subgraphs of hyperbolic tilings for large values of~$q$?
\item Can the techniques be extended to \textsc{Subgraph Isomorphism} in Gromov-hyperbolic planar graphs?
\item Is there a $n^{\Omega(\log n)}$ lower bound for recognizing subgraphs, of hyperbolic tilings, or is it polynomial-time solvable? The only known quasi-polynomial lower bound technique for hyperbolic settings~\cite{Kisfaludi-Bak20} does not apply here, as it relies on locally dense structures.
\item Are there $2^{o(n)}$ algorithms for recognizing subgraphs of higher dimensional regular tilings?
\end{itemize}

\bibliography{bibliography}
\clearpage

\appendix

\section{Missing proof of \Cref{sec:prelims}}

\boundary*

\begin{proof}
    The proof of Lemma 5.5 of~\cite{KisfaludiBak25} relies on Lemma 5.3 of~\cite{KisfaludiBak25}, where it is proved that the unbounded face (henceforth called \EMPH{boundary}) of a convex hull has size $O(N)$, where $N$ has a different definition than our graph size $n$. 
     Thus, in order to use the proof of Lemma 5.5 of~\cite{KisfaludiBak25} in our setting, we only need to show that the boundary of $\conv(H)$ has $O(n)$ vertices.    
    Let $u,v$ be two vertices of $H$ that are on the boundary of $\conv(H)$.
    The length of the shortest path on the boundary of $\conv(H)$ is at most the length of the shortest path from $u$ to $v$ following edges of $H$. Now, the length of the boundary of $\conv(H)$ is upper bounded by
    the sum of shortest paths following edges of $H$ between vertices of $H$ that are on the boundary of $\conv(H)$.
    As $H$ is a subgraph of $T_{p,q}$, those shortest paths are on the unbounded face of $H$. 
    Each vertex on the boundary of $H$ is intersected at most twice, so the boundary of $\conv(H)$ has length at most $2n$. 
\end{proof}

\section{Full details of \Cref{sec:Struct}}

The following proof was missing from the main paper.

\bisectorray*

\begin{proof}
    The ray $r$ intersects some vertices of $\til$. In addition to that, depending on the parity of $p$, the ray $r$ may also intersect edges of $\til$ at the edge's middle point, or intersect edges of $\til$ such that $r$ covers the edge entirely. 
    We assume for the sake of contradiction that 
    some vertex $w$ of the convex hull of $H$, distinct from $v$, and intersected by $r$.
    According to \Cref{lm:kisfaudibak25_lemma3_3}, any shortest path between $v$ and $w$ visits, on the segment between $v$ and $w$, (i) all vertices intersected by $r$, and (ii) for each edge cut by $r$, at least one vertex incident to that edge.  
    Because $v,w\in V(\conv(H))$, all vertices visited by any shortest path between $v$ and $w$ are in $V(\conv(H))$. 

    If $p$ is even, the first element that $r$ intersects after $v$, is the vertex of tile $\tau$ opposite to $v$. We denote that vertex $w'$. Now, according to \Cref{lm:path_around_tile} all the paths between $v$ and $w'$ contain only vertices of $V(\tau)$. Because of symmetry, the shortest path from $v$ to $w'$ on either side of $\tau$ have the same length. So, all vertices in $V(\tau)$ are vertices in $V(\conv(H))$. Thus, $\tau$ is an bounded face of $\conv(H)$, which is a contradiction.
    
    If $p$ is odd, the first element of $\til$ that $r$ intersects after $v$ is the middle point of the edge of $\tau$ opposite to $v$. We denote that edge $e$. There is exactly one other tile incident to $e$, we denote that tile $\tau'$. The second element of $\til$ intersected by $r$ after $v$ is the vertex of $\tau'$ opposite to $e$, that we denote $w''$. Now, according to \Cref{lm:path_around_tile} and \Cref{lm:kisfaudibak25_lemma3_3}, any shortest path between $v$ and $w''$ visits only vertices of $V(\tau\cup\tau')$. Because of symmetry, paths on either side of $\tau\cup\tau'$ have the same length, thus, all vertices in $V(\tau\cup\tau')$ are vertices in $V(\conv(H))$. Thus, $\tau$ is an bounded face of $\conv(H)$, which is again a contradiction and completes the proof.
\end{proof}

\subsection{Normalized nooses.}

We give here the full description of our normalized nooses.
We can map a sphere cut decomposition from the sphere $\mathbb{S}^2$ to the hyperbolic plane $\mathbb{H}^2$ by puncturing one point $p$ on the sphere that will become the ideal boundary of $\mathbb{H}^2$. We call the punctured point $p$ in $\mathbb{S}^2$ the \EMPH{point at infinity}. 
Each curve in $\ma S^2$ from a point $a$ to a point $b$ through the point at infinity $p$ maps to an infinite length curve in $\mathbb{H}^2$ that goes from the image of $a$ to the ideal boundary of $\mathbb{H}^2$, then the curve can have an ideal segment on the ideal boundary, and finally, has a segment from the ideal boundary of $\ma H^2$ to the image of $b$. 

We define a normalized geometric realization for our nooses. We first define normalized curves, which is a restricted class of curves, then we define normalized nooses, which are a noose that are normalized curves.
A \EMPH{normalized curve}, or \EMPH{generalized polygon}~\cite{Blasius24} in $\mathbb{H}^2$ is a plane cycle, where each vertex is either a vertex of $\til$, the center of a tile in $\til$, or a point in the unbounded face, which can be an ideal point on the ideal boundary of $\ma H^2$.
On the finite part $\gamma^-$ of a normalized curve $\gamma$, that is, the subcurve of $\gamma$ that contains only vertices of $\til$ and tile centers, 
the vertices of $\gamma^-$ alternate between vertices of $\til$ and tile centers. More precisely, neighbors of a vertex $v$ of $\til$ are tile centers of tiles adjacent to $v$. And, neighbors of a tile center $t$ of a tile $\tau$ are vertices of $\til$ that are also vertices of $\tau$. Furthermore, a tile center cannot be a neighbor of an ideal vertex.
An edge between a tile center and vertex of $\til$ is a straight line segment, and an edge between a vertex of $\til$ and an ideal point is a ray, and an edge between two ideal points is an ideal arc. 
A \EMPH{normalized noose} of a convex subgraph $H$ of $\til$ is a normalized curve with the following properties:
\begin{enumerate}
    \item Let $\tau$ be a tile of $\tiling$, the vertices $u,v\in V(\tau)$ and the point $t$ be the center of $\tau$. For a noose segment $\nu\cap\tau$ of $\nu$, from $v_1$ to $v_2$ through $\tau$, $ \nu\cap\tau$ consists of the two straight line segments from $v_1$ to $t$ and from $t$ to $v_2$. Note that according to \Cref{lm:single_tile}, all bounded faces of $H$ are tiles of $\tiling$.
    \item For a noose segment $ \nu_\infty$ through the unbounded face of $H$ from a vertex $w_1$ to a vertex $w_2$, let $r_1$ be the angle bisector ray to infinity starting from $w_1$ and through the leftmost tile of $\tiling$ that is adjacent to $w_1$ and on the unbounded face of $H$. Note that according to \Cref{lm:angle_bisector_ray}, $r_1$ is disjoint from $H\setminus w_1$. Similarly, let $r_2$ be the angle bisector ray to infinity from $w_2$ through the leftmost tile of $\tiling$ that is adjacent to $w_2$ and on the unbounded face of $H$. We call $i_1$ and $i_2$ the ideal endpoints of $r_1$ and $r_2$ respectively.
    If $r_1$ and $r_2$ intersect at a point $p$, we call the straight line segments from $w_1$ to $p$ and from $p$ to $w_2$. 
    If $r_1$ and $r_2$ do not intersect, $ \nu_\infty$ consists of $r_1$ followed by an ideal arc at infinity between $i_1$ and $i_2$ such that the ideal arc turns right after $r_1$ according to the orientation of $\nu$, followed by $r_2$.
\end{enumerate}

\subsection{Encoding candidate nooses.}

Each vertex in $\til$ has $q$ neighboring tiles and each tile center has $q$ vertices of $\til$ in its tile. 
So, for a noose $\nu$ and a vertex $v\in \nu\cap\til$, there are $q$ possible tile centers that can be the next vertex of $\nu$. After that tile center, there are $p$ vertices of $\til$ that can be the next vertex of $\nu$. So, there are $pq$ possible noose segments from $v$ to the next vertex in $\nu\cap\til$.
Now, we encode a noose $\nu$ such that we fix a distinguished origin vertex, and 
we first encode a shortest path from the origin to any vertex $v$ of $\nu$ as a sequence of directions such that each step tells us in which direction we will go next.
Then, we encode $\nu$ as a sequence of directions starting from $v$ and ending in $v$. 
We define the directions at every vertex such that, the first time we visit a vertex, we label all  directions from that vertex with a label $j$ where $0\le j\le q-1$ with a fixed cyclic order and direction 0 is the direction where we came to this vertex. 
If $\nu$ intersects the origin, we will select $v$ to be the origin, so the path to $v$ will have length 0.

Given any sequence of directions, one can simulate the corresponding walk in $\til$ and decide whether it forms a closed cycle in time polynomial in the size of the sequence~\cite{Bridson09,Epstein92}.

To encode a normalized noose $\nu'$, where $\nu'\not=(\nu')^-$, we first encode a shortest path $P$ from the origin to a specific vertex $v_1\in V((\nu')^-)$, where $v_1$ is the first vertex in the sequence of vertices of the open subcurve $(\nu')^-$ ordered according to the orientation of $\nu'$. After $P$, we encode the open subcurve $(\nu')^-$ as a sequence of directions. Finally, we encode in which of the $q$ directions the rays at both ends of $(\nu')^-$ point. The rays can be connected in only one way whether they intersect or not, so the ideal arc, or the point of intersection of rays does not need to be encoded.

We encode a candidate noose $\nu''$ the same way as a normalized noose, with the two following restrictions:
\begin{itemize}
    \item the path from the origin to the first vertex of $\nu''$ is always empty, and
    \item the direction of the first edge in any candidate noose is fixed.
\end{itemize}

\end{document}